\documentclass[10pt,journal]{IEEEtran}
\ifCLASSINFOpdf
\else
\fi
\usepackage{graphicx}
\usepackage{amsmath}
\usepackage{amsthm}
\usepackage{cite}
\usepackage{amsfonts}
\usepackage{float}
\usepackage{mathtools}
\usepackage{tikz}
\usepackage{amssymb}

\usepackage{algorithm}
\usepackage[noend]{algpseudocode}

\newtheorem{theorem}{Theorem}

\newcommand{\hvec}[1]{\ensuremath{\Hat{\vec{#1}}}}
\newcommand{\bvec}[1]{\ensuremath{\Breve{\vec{#1}}}}
\renewcommand{\vec}[1]{\ensuremath{\boldsymbol{#1}}}
\newcommand{\mc}[1]{\ensuremath{\mathcal{#1}}}
\newcommand{\of}[1]{^{(#1)}}

\renewcommand{\eqref}[1]{(\ref{eq:#1})}

\newcommand{\secref}[1]{Section~\ref{sec:#1}}
\newcommand{\appref}[1]{Appendix~\ref{app:#1}}

\begin{document}
%
\title{Generalized Sparse Regression Codes \\ for Short Block Lengths}

\author{
\IEEEauthorblockN{Madhusudan Kumar Sinha \IEEEauthorrefmark{1}, Arun Pachai Kannu\IEEEauthorrefmark{2}}\\
  \IEEEauthorblockA{Department of Electrical Engineering \\
		    Indian Institute of Technology Madras \\
		    Chennai - 600036, India\\
		    Email: \IEEEauthorrefmark{1}ee16d028@ee.iitm.ac.in, \IEEEauthorrefmark{2}arunpachai@ee.iitm.ac.in }
}


%


\maketitle

\begin{abstract}

Sparse regression codes (SPARC) connect the sparse signal recovery framework of compressive sensing with error control coding techniques. SPARC encoding produces codewords which are \emph{sparse} linear combinations of columns of a dictionary matrix. SPARC decoding is accomplished using sparse signal recovery algorithms. We construct dictionary matrices using Gold codes and mutually unbiased bases and develop suitable generalizations of SPARC (GSPARC). We develop a greedy decoder, referred as match and decode (MAD) algorithm and  provide its analytical noiseless recovery guarantees. We propose a parallel greedy search technique, referred as parallel MAD (PMAD), to improve the  performance. We describe the applicability of GSPARC with PMAD decoder for multi-user channels, providing a non-orthogonal multiple access scheme. 
We present numerical results comparing the block error rate (BLER) performance of the proposed algorithms for GSPARC in AWGN channels, in the short block length regime. The PMAD decoder gives better BLER than the approximate message passing decoder for SPARC. GSPARC with PMAD gives comparable and competitive BLER performance, when compared to other existing codes. In multi-user channels, GSPARC with PMAD decoder outperforms the sphere packing lower bounds of an orthogonal multiple access scheme, which has the same spectral efficiency.


 
 
\end{abstract}



%
\IEEEpeerreviewmaketitle

\section{Introduction}

Shannon's seminal paper on information theory established the existence of information encoding and decoding techniques that guarantee almost error-free communications across noisy channels \cite{shannon1948mathematical}. Extensive work has been carried out to develop such efficient error control coding techniques for additive white Gaussian noise (AWGN) channels \cite{lin2001error}. 
Error correcting codes such as turbo codes and LDPC codes 
are widely used in various communication systems today, which have very small block error rates for large block lengths.

In our work, we consider error control coding in the short block length regime.  In many communication systems, the control channel information is typically sent over short block lengths. Several use cases in the fifth generation (5G) of mobile networks such as ultra-reliable low-latency communication (URLLC) and massive machine type communications in IoT applications require short block lengths. For instance, in URLLC scenarios such as industrial automation, autonomous vehicles and augmented/virtual reality, short length codes are required to meet the low latency requirements.  A comprehensive study on the performance of existing codes in the short block length regime has been done in  \cite{van2018short,cocskun2019efficient}.

Sparse regression codes (SPARC) \cite{joseph2012least,joseph2013fast} connect the sparse signal recovery framework of compressive sensing \cite{eldar2012compressed} with error control coding techniques. In SPARC, a dictionary matrix (design matrix) $\vec{A}$ of size $N \times L$ with $L >> N$ is partitioned into $K$ equal sub-blocks (sections), with each sub-block having $\frac{L}{K}$ columns. Based on the information bits, one column is chosen from each block and the codeword for transmission is obtained as the sum of chosen columns. We can represent the codeword as $\vec{s} = \vec{A} \vec{x}$, where $\vec{x}$ is a $K$-sparse signal with exactly $K$ non-zero entries.  The non-zero entries of $\vec{x}$ are fixed and known in the standard SPARC \cite{rush2017capacity} and they are chosen from a PSK constellation (based on information bits) in the modulated SPARC \cite{hsieh2021modulated}. 

In \cite{joseph2012least,joseph2013fast}, the SPARC codes using Gaussian dictionary matrices were proven to achieve channel capacity for AWGN channels, as the block lengths approach infinity. 
Several power allocation (across the sub-blocks) and spatial coupling techniques have been developed for SPARC  
\cite{joseph2013fast,cho2013approximate,rush2017capacity,hsieh2021modulated, greig2017techniques,barbier2015approximate,barbier2017approximate,barbier2019universal} to improve the empirical performance of SPARC codes.
In \cite{rush2017capacity,hsieh2021modulated} approximate message passing (AMP) decoders have been developed for standard and modulated SPARC, which guarantee that sub-block error rate goes to zero for all rates below capacity in AWGN channels. 
It has been shown empirically that the AMP decoder derived for Gaussian dictionary matrices work well with other dictionary matrices like the ones based on Hadamard matrices \cite{barbier2017approximate,rush2017capacity} and fast Fourier transform (FFT) matrices \cite{hsieh2021modulated}. Clipping and generalized AMP are discussed in \cite{liang2021finite}, in order to improve finite block length performance of SPARC at low to medium code rates (bits per channel use). Iterative power allocation techniques are given in \cite{greig2017techniques}, which improve the performance of SPARC in high code rates.

In our work, we consider SPARC for short block lengths, with $N \leq 128$ and code rate $\approx0.5$ bits per channel use (bpcu), a regime which has gained sufficient interest in the recent years \cite{van2018short,cocskun2019efficient} and where the methods to improve finite length performance of SPARC doesn't work \cite{liang2021finite,greig2017techniques}. We construct good deterministic dictionary matrices, utilizing the  existing literature on generating a large set of sequences with good correlation properties. Specifically, we use Gold codes \cite{Goldcode} from the CDMA literature and  mutually unbiased bases (MUB) \cite{wootters1989optimal} from the quantum information theory. With these constructions, the number of columns in $\vec{A}$ is $L \approx N^2$, where $N$ is the length of each column and the maximum (normalized) cross correlations among the columns (usually referred as mutual coherence in compressive sensing) is approximately  $\frac{1}{\sqrt{N}}$. We choose the sparsity level $K$ close to $\sqrt{N}$, which typically ensures that the bpcu falls in the regime of interest when $N \leq 128$. The contributions of our work are summarized below.

We construct dictionary matrices using Gold codes and mutually unbiased bases, which have not been previously used in SPARC.

We generalize the SPARC by possibly allowing sub-blocks of different sizes, which we refer as sub-block structure encoding (SSE) scheme. We also allow modulation of the $K$ selected columns using information symbols from finite alphabet constellations. We give an algorithm to partition the given dictionary matrix with a total of $L$ columns into $K$ sub-blocks so that size of each sub-block is power of 2 and the number of information bits carried by SSE is maximized.

We also generalize SPARC by entirely eliminating the sub-block structure, which we refer as sub-block free encoding (SFE) scheme. In SFE, we allow choosing \emph{any} $K$ columns from a total of $L$ columns from the dictionary matrix. We give an iterative procedure which uniquely maps the information bits to one of the $\binom{L}{K}$ combinations. Our iterative procedure is quite efficient and eliminates the need for any look-up tables.

We develop a simple greedy algorithm for the generalized SPARC (GSPARC), which we refer as match and decode (MAD) algorithm, to recover the sparse signal (and subsequently the information bits) from the noisy observation of the codeword. Our MAD algorithm inherently exploits the finite alphabet nature of the modulation symbols and performs better than the conventional orthogonal matching pursuit (OMP) algorithm in AWGN channels. We give analytical recovery guarantees of the MAD decoder, in terms of the coherence of the dictionary matrix and the coherence parameter of the modulating constellation symbols.

We improve the MAD algorithm by introducing a parallel search mechanism, which we refer as parallel MAD (PMAD) algorithm. Our MAD and PMAD decoders do not require the knowledge of the channel noise variance. Using numerical simulations, we show that our PMAD algorithm performs better than the AMP algorithm \cite{hsieh2021modulated} in AWGN channels, for short block lengths.  

We also show that our PMAD with GSPARC provides competing block error rate performance in the short block lengths, when compared with several existing codes \cite{cocskun2019efficient,van2018short}. 
In  addition, we also show that SSE can be used in multi-user channels, such as multiple-access, broadcast and interference channels. 
For some combinations of code rates, block lengths and number of users, we show that SSE with PMAD decoder outperforms the sphere packing lower bounds of an orthogonal multiple access scheme.

The paper is organized as follows: In \secref{enco}, we provide the details of the encoding techniques for GSPARC. In \secref{dict}, we give the details of dictionary matrix construction. In \secref{deco}, we describe the decoding algorithms and their analytical performance guarantees. 
In \secref{mult}, we discuss on how SSE and PMAD can be employed in multi-user communication channels. In \secref{simu}, we present block error rate performance comparison in AWGN channels. In \secref{conc}, we present conclusions and give directions for future work.

\section{Generalized SPARC Encoding Procedure} \label{sec:enco}

\subsection{Sub-block Structure Encoding} \label{sec:sse}
Consider a dictionary matrix $\vec{A}$ of size $N \times L$, with unit norm columns and $L \geq N$.  
The codewords for messages are obtained using \emph{sparse} linear combinations 
of columns of the matrix $\vec{A}$. In SSE, we fix the \emph{sparsity} level as $K$ with $K \leq N$ and partition the dictionary matrix $\vec{A}$ into $K$ subblocks (also referred as sections) such that $\vec{A} = \left[\vec{A}_1 \cdots \vec{A}_K \right]$ with $k^{th}$ sub-block $\vec{A}_k$ having a size of $N \times L_k$ and $\displaystyle \sum_{k=1}^K L_k = L$. We assume that the number of columns in each sub-block is a power of 2 and the sub-blocks can possibly have unequal sizes. Based on the information bits, one column from each sub-block is selected and transmit codeword is obtained as a linear combination
\begin{align}
\vec{s} &= \sum_{k=1}^K \beta_k \vec{a}_{\alpha_k}, \label{eq:cw1}
\end{align}
where $\vec{a}_{\alpha_k}$ is a column from sub-block $\vec{A}_k$ and the \emph{modulation symbol} $\beta_k$ is chosen from an $M$-ary constellation $\mc{M}$. Now, the codeword in \eqref{cw1} can be represented as
\begin{align}
\vec{s} &= \vec{A}\vec{x}, \label{eq:cw2}
\end{align}
where $\vec{x}$ of size $L \times 1$ is a sparse signal with only $K$ non-zero entries from the constellation $\mc{M}$. We also allow the special case of $M=1$, for which $\beta_k = +1, \forall k$. Let $\mc{S}$ denote the support of $\vec{x}$. Since the symbol $\beta_k$ carries $\log M$ bits (assuming $M$ is a power of 2) and the columns of sub-block $\vec{A}_k$ are indexed using $\log L_k$ bits, the total number of information bits encoded in the codeword in \eqref{cw1} is
\begin{align}
N_b &= K \log M + \sum_{k=1}^K \log L_k. \label{eq:nbsse}
\end{align}
We have used base 2 for $\log$ throughout the paper. We define the \emph{code rate} of the encoding scheme in units of bits per real channel use (bpcu) as the number of bits transmitted per  
real dimension utilized. If $\vec{A}$ is a real matrix and the constellation $\mc{M}$ is real (such as PAM, BPSK), 
the code rate is $\frac{N_b}{N}$ bpcu.  On the other hand, if $\vec{A}$ is a complex matrix and/or the constellation symbols are complex, 
the code rate is $\frac{N_b}{2N}$ bpcu. SPARC encoding in \cite{hsieh2021modulated,rush2017capacity} mandates all the sub-blocks to be of equal sizes. Since we allow for unequal sub-block sizes, SSE scheme \eqref{cw1} is a generalization of the SPARC encoder. 


\subsection{Sub-block Partitioning Algorithm} \label{sec:spa}

Deterministic construction of sequences with good correlation properties exist in the literature of CDMA \cite{Goldcode,frank1963polyphase,chu1972polyphase}  and quantum information theory \cite{wootters1989optimal,renes2004symmetric}. 
Dictionary matrices based on these deterministic constructions are good candidates due to their small coherence values. However, these constructions exist only for certain values of $N$ and $L$. In these cases, we need to have a proper sub-block partitioning algorithm such that the number of information bits \eqref{nbsse} conveyed through the codeword \eqref{cw1} is maximized for a given $K$. 

Given the total number of columns $L$ in the dictionary matrix and the required number of partitions $K$, we want to optimize 
$\sum_{k=1}^K \log(L_k)$  with the constrain that each $L_k$ is a power of $2$ and $\displaystyle \sum_{k=1}^K L_k \leq L$. If we allow $L_k$ to take any real value, the solution for the above optimization problem is readily obtained as $L_1=L_2=\cdots=L_K=\frac{L}{K}$, that is, all the sub-blocks should be of equal size. To meet the power of 2 constraint, we set the size of the smallest sub-block as $\displaystyle L_1 = 2^{\lfloor \log\frac{L}{K} \rfloor}$, the largest power of 2 number which is less than or equal to $\frac{L}{K}$. After this step, the problem reduces to divide $L-L_1$ columns into $K-1$ sub-blocks. Proceeding in the same manner iteratively, the optimal partitioning sizes are obtained as 
\begin{align}
L_k &= 2^{\lfloor \log \frac{L-\sum_{m=1}^{k-1}L_m}{K-k+1} \rfloor},~~ k=1,\cdots,K.
\end{align}   
For example, if $L=23$ and $K=3$, the optimal partition sizes are $L_1=4$, $L_2=L_3=8$ and the remaining $3$ columns are unused/discarded. From the above procedure, it also follows that the size of the largest sub-block can be at most twice the size of the smallest sub-block, that is, $L_K \leq 2 L_1$.

\subsection{Sub-block Free Encoding}

In SFE, we eliminate the sub-block structure and choose \emph{any} subset of $K$ columns from a total of $L$ columns and modulate 
the chosen columns using symbols from an $M$-ary constellation. SFE scheme is another generalization of the SPARC encoding scheme from \cite{hsieh2021modulated,rush2017capacity}. The number of bits encoded by SFE scheme will be
\begin{align}
N_b &= K \log M + \lfloor \log \binom{L}{K} \rfloor, \label{eq:nbsc}
\end{align} 
which will be larger than or equal to that of the SSE scheme \eqref{nbsse}. Unlike SSE scheme, mapping bits into a subset of $K$ columns is not straightforward. We provide an iterative scheme to achieve this feat without using any look-up table. 

Specifically, we provide a one-to-one mapping between non-negative integers and combinations of $K$ objects out of $L$ objects. In our SFE-GSPARC, a sequence of $N_b$ bits (representing a non-negative integer) is mapped to a unique combination of $K$ columns out of the total $L$. We use lexicographic ordering within each combinations to form a unique ordered set (word) representing a combination. We then use lexicographic ordering over all possible words to list all possible combinations. This ordering allows us to get a one-to-one mapping between bits and object combinations. We provide a numerically efficient method to map bits to combinations and vice versa, without generating and storing the actual list.

Let $\mc{B}=\{0,1,2,3,...,L-1\}$ be a set of $L$ distinct objects, where we use the first $L$ non-negative integers as an abstraction for a set of $L$ different objects. The total number of combinations of $K$ objects out of $L$ objects is given by $\binom{L}{K}$. Let a combination be represented uniquely by the ordered set $\vec{b}=(b_0,...,b_{K-1})$ such that $0\leq b_0<b_1<...<b_{K-1}\leq L-1$. The lexicographic ordering on the ordered set representation of the combinations allow us to list the combinations against non-negative integers.  For example, with $L=5$ and $K=3$, there are $\binom{5}{3}=10$ unique combinations. The lexicographic listing of these combinations against non-negative integers is given in Table~\ref{example:5C3}. 
\begin{table}[h]
\centering
\begin{tabular}{|l|l|}
\hline
index & combinations \\ \hline
0 &  (0,1,2)      \\ \hline
1 &  (0,1,3)      \\ \hline
2 &  (0,1,4)      \\ \hline
3 &  (0,2,3)      \\ \hline
4 &  (0,2,4)      \\ \hline
5 &  (0,3,4)      \\ \hline
6 &  (1,2,3)      \\ \hline
7 &  (1,2,4)      \\ \hline
8 &  (1,3,4)      \\ \hline
9 &  (2,3,4)      \\ \hline
\end{tabular}
\caption{Lexicographic listing of 3 objects chosen out of 5.}
\label{example:5C3}
\end{table}

We note that, out of $\binom{L}{K}$ combinations, $\binom{L-i}{K}-\binom{L-(i+1)}{K}=\binom{L-(i+1)}{K-1}$ combinations start with object $i$ where $i\in\{0,1,...,L-K\}$. The decimal indices of combinations starting with object $i$ starts at $\binom{L}{K}-\binom{L-i}{K}$ and ends at $\binom{L}{K}-\binom{L-(i+1)}{K}-1$. For the decimal index $d \in \{0,\cdots,\binom{L}{K}-1\}$ represented by the combination $\vec{b}=(b_0,b_1,...,b_{K-1})$, we will have $b_0 = i$, if 
\begin{align}
\binom{L}{K}-\binom{L-i}{K} &\leq d <\binom{L}{K}-\binom{L-(i+1)}{K}. \label{eq:require}
\end{align} 
In  \appref{index_comb}, we show that, non-negative integer $i$ satisfying above constraint is upper bounded as
\begin{align}
    i &\leq \bigg\lfloor
    \left(L-\frac{(K-1)}{2}\right)\left\{
    1-\left(1-\frac{d}{\binom{L}{K}}\right)^\frac{1}{K}\right\}\bigg\rfloor ~~:= \bar{i}(L,K). \label{eq:ibar}
\end{align}
Given $L$ and $K$, we first find $\bar{i}(L,K)$ from \eqref{ibar} and check if $\bar{i}$ satisfies \eqref{require} for the given $d$. If not, we keep decrementing $\bar{i}$ by one, until we find the integer $i$ satisfying the constraint \eqref{require}. Once the first object $b_0$ is chosen, the problem is reduced to choosing $K-1$ objects out of $L-b_0$ objects corresponding to the decimal index $\displaystyle d-\left[\binom{L}{K}-\binom{L-b_0}{K}\right]$. Hence, the same procedure can be recursively applied until the last object $b_{K-1}$ is chosen. In our simulations, we find that either $\bar{i}$ in \eqref{ibar} or $\bar{i}-1$ always satisfies the requirement in \eqref{require}.

Using the same counting argument, given the set of $K$ objects $\{b_0,\cdots,b_{K-1}\}$ (from a total of $L$), we can find the decimal index corresponding to the lexicographic ordering as
\begin{align*}
    d &=\binom{L}{K}-\sum_{k=0}^{K-2}\left[\binom{L-b_k}{K-k}-\binom{L-b_{k}-1}{K-k-1}\right] - \binom{L-b_{K-1}}{1}, \\
     &=\binom{L}{K} - \left[\sum_{k=0}^{K-2}\binom{L-b_k-1}{K-k}\right] - \binom{L-b_{K-1}}{1}.
\end{align*}



\section{Dictionary Matrix Construction} \label{sec:dict}

The choice of the dictionary matrix $\vec{A}$ plays a vital role in the block error performance. It is desirable that the dictionary matrix has a large number of columns (as the number of information bits increases with $L$, for fixed $N$ and $K$) with small correlation among the columns (for good sparse signal recovery performance), which is characterized by the the mutual coherence of the dictionary matrix $\vec{A}$, defined as,
\begin{equation}
\mu(\vec{A}) = \max_{p \neq q}  \frac{|\langle \vec{a}_p, \vec{a}_q \rangle|}{\|\vec{a}_p\| \|\vec{a}_q\|}.
\end{equation}
In this paper, we consider  dictionary matrix constructions using Gold code sequences from CDMA literature and mutually unbiased bases from quantum information theory, which have small coherence values.

\subsection{Gold Codes}
Gold codes are binary sequences with alphabets $\{\pm 1\}$. Considering lengths of the form $N = 2^n-1$, where $n$ is any positive integer, there are  $2^n+1$ Gold sequences. By considering all the circular shifts of these sequences, we get $2^{2n}-1$ sequences. When dictionary matrix columns are  constructed with these $2^{2n}-1$ sequences normalized to unit norm, the resulting cross-correlation between any two columns of the dictionary matrix matrix takes only three possible values given as $\frac{-1}{N}$, $\frac{-t(n)}{N}$ and $\frac{t(n)-2}{N}$ where $t(n)$ is given by \cite{Goldcode}, 
\begin{equation}
    t(n)=\begin{cases} 
      1+2^\frac{n+1}{2}, & n \text{ is odd,} \\
      1+2^\frac{n+2}{2}, & n \text{ is even.}
   \end{cases} 
\end{equation}
The mutual coherence of the gold code dictionary matrix is thus given by
\begin{equation}
   \mu=\frac{t(n)}{N} \label{eq:mugold}
\end{equation}
and we note that odd value of $n$ leads to smaller values of mutual coherence.  We can add any column of the identity matrix to the Gold code dictionary matrix, to get a total of $L=2^{2n}$ columns (which is a power of 2), with the mutual coherence same as \eqref{mugold}. 
Storing such a dictionary matrix will require $N(N+1)^2$ bits. 

\subsection{Mutually Unbiased Bases}
Two orthonormal bases $\mc{U}_1$ and $\mc{U}_2$ of the $N$-dimensional inner product space $\mathbb{C}^N$ are called mutually unbiased if and only if $|\langle\vec{x},\vec{y}\rangle|=\frac{1}{\sqrt{N}}$ for any $\vec{x}\in \mc{U}_1$ and $\vec{y}\in \mc{U}_2$. A collection of orthonormal bases of $\mathbb{C}^N$ is called mutually unbiased if all bases in the collection are pairwise mutually unbiased.
Let $Q(N)$ denote the maximum number of orthonormal bases of $\mathbb{C}^N$, which are pairwise mutually unbiased. In \cite{wootters1989optimal}, it has been shown  that $Q(N) \leq N$ (excluding the standard basis), with equality if $N$ is a prime power. Explicit constructions are also given in \cite{wootters1989optimal} for 
getting $N$ MUB in $N$-dimensional complex vector space $\mathbb{C}^N$, if $N$ is a prime power.

When $N=2^n$, with $n \geq 2$,  the $N$ MUB unitary matrices $\{\vec{U}_1,\cdots,\vec{U}_N\}$ have the following properties. 
\begin{itemize}
\item 

The entries in all the $N$ unitary matrices belong to the set $\{\frac{+1}{\sqrt{N}},\frac{-1}{\sqrt{N}},\frac{+j}{\sqrt{N}},\frac{-j}{\sqrt{N}}\}$. This  follows from the construction of MUB given in \cite{wootters1989optimal}. Storing all these $N$ unitary matrices will require $2 N^3$ bits. 
\item 

For $N$ up to 512, we find that the inner products  $\langle\vec{x},\vec{y}\rangle$ between $\vec{x}\in \vec{U}_i$ and $\vec{y}\in \vec{U}_m$ for $i \neq m$ is given by
\begin{equation}
    \langle\vec{x},\vec{y}\rangle\in\begin{cases} 
      \{\frac{1}{\sqrt{N}}e^\frac{j m 2\pi}{8}:m=0,...,7\}, & n \text{ is odd,} \\
      \{\frac{1}{\sqrt{N}}e^\frac{j m 2\pi}{4}:m=0,...,3\}, &   n \text{ is even.}
   \end{cases} 
   \label{eq:innerProdsMUB}
\end{equation}
We conjecture that this property holds true when $N$ is any higher power of 2. 
\end{itemize}

We construct dictionary matrix using $N$ MUB as $\vec{A} = [\vec{U}_1 \cdots \vec{U}_N]$. In this case, $L = N^2$ and the corresponding mutual coherence $\mu$ is $\frac{1}{\sqrt{N}}$. In addition, when $N$ is a power of $2$, we can always partition $\vec{A}$ into $K$ sub-blocks with size of each sub-block $L_k$ being a power of $2$ and each $L_k \geq \frac{N^2}{2K}$. 

\section{Decoding Algorithms} \label{sec:deco}

\subsection{Match and Decode Algorithm}

The received signal is modeled as
\begin{eqnarray}
\vec{y} &=& \vec{s} + \vec{v}, \nonumber \\ 
&=& \vec{A}\vec{x} + \vec{v}, \label{eq:obs1}
\end{eqnarray}
where $\vec{v}$ is additive noise. Information bits can be retrieved by recovering the sparse signal $\vec{x}$ from the observation $\vec{y}$. Conventional sparse signal recovery can be done using greedy techniques \cite{mallat1993matching,cai2011orthogonal,tropp2004greed} or convex programming based techniques \cite{chen2001atomic} or iterative message passing techniques \cite{beck2009fast,messageMontanari}. However, with our SPARC encoding, $\vec{x}$ has special structure. The non-zero entries of $\vec{x}$ are from a finite alphabet constellation. In addition, for the SSE scheme, there is exactly one non-zero entry corresponding to each sub-block. Such structures need to be utilized in order to provide good error performance.  Approximate message passing decoders which exploit the structure of the SPARC signal $\vec{x}$ are developed in \cite{hsieh2021modulated,rush2017capacity}, and their sub-block error rate asymptotically (as $N$ and $L$ grow to $\infty$) converges to zero for AWGN channels, for all rates below channel capacity.

In this paper, we develop a simple greedy decoder, referred as match and decode algorithm, which utilizes the structure in the SPARC signal $\vec{x}$. MAD algorithm for SSE and SFE is described in Algorithm~\ref{mad}. We would like to emphasize that our MAD algorithm does not need to know any noise statistics, such as its variance. MAD algorithm takes the dictionary matrix $\vec{A}$, the observation $\vec{y}$, sparsity level $K$ as inputs and produce an estimate $\hvec{x}\of{K}$ of the sparse signal $\vec{x}$. It is ensured that the estimate $\hvec{x}\of{K}$ (of size $L$) has exactly $K$ non-zero entries from the constellation set $\mc{M}$. Any sparse signal $\hvec{x}$ (of size $L$) with at most $K$ non-zero entries from the set $\mc{M}$ can also be given as partial information to the MAD algorithm. If no partial information is available, $\hvec{x}$ is set as $\vec{0}$.   

\begin{algorithm}
\caption{Match and Decode Algorithm}\label{mad}
\begin{algorithmic}[1]

\State \textbf{Input:} Get the observarion $\vec{y}$, dictionary matrix $\vec{A}$, sparsity level $K$ and any partially recovered sparse signal  $\hvec{x}$ with support $\mc{S}_{\hat{x}}$ with $|\mc{S}_{\hat{x}}| < K$ . 

\State \textbf{Initialize:}  
Initialize the iteration counter $t = |\mc{S}_{\hat{x}}|$, the residual $\vec{r}\of{t} = \vec{y} - \vec{A} \hvec{x}$ and 
the estimate $\hvec{x}\of{t} = \hvec{x}$. 
Let $\hat{\mc{S}}\of{t}$ denote the set of columns in the dictionary matrix discarded by the algorithm (based on the detected ones) until the $t^{th}$ iteration. 
If $\hvec{x}=\vec{0}$, then $\hat{\mc{S}}\of{0} = \emptyset$. For SFE scheme, initialize 
$\hat{\mc{S}_t} = \mc{S}_{\hat{x}}$. For SSE scheme, $\hat{\mc{S}}\of{t} = \cup_{i \in \mc{S}_{\hat{x}}} \vec{A}_{k(i)}$, where $k(i)$ corresponds to the sub-block $k$ which contains $i^{th}$ column of the dictionary matrix $\vec{A}$.  

\State \textbf{Match:} Correlate the residual with the columns of the dictionary matrix and the constellation symbols as given below. 
\begin{align}
c_i  &= \langle \vec{r}\of{t} , \vec{a}_i \rangle, ~~ i \in \{1,\cdots,L\} \setminus \hat{\mc{S}}\of{t} \label{eq:metri0} \\
p_{i,m} &= \mathfrak{Real}\{c_i b_m^*\} - \frac{|b_m|^2}{2}, ~~ b_m \in \mc{M} \label{eq:metri}
\end{align}

\State \textbf{Decode:} Detect the active column and the corresponding modulation symbol as, 
$(\hat{i},\hat{m}) = \arg\max_{\substack{i \in \{1,\cdots,L\} \setminus \hat{\mc{S}}\of{t} \\ 1 \leq m \leq M}} p_{i,m}$

\State \textbf{Update:} Update the recovered sparse signal information $\hvec{x}\of{t+1} = \hvec{x}\of{t} + b_{\hat{m}} \vec{e}_{\hat{i}}$. (Here $\vec{e}_n$ denotes $n^{th}$ standard basis of size $L$.) Update the residual  
\begin{align}
\vec{r}\of{t+1} &=   \vec{r}\of{t} -   b_{\hat{m}} \vec{a}_{\hat{i}}. \label{eq:resi} 
\end{align}
For SFE scheme, update $\hat{\mc{S}}\of{t+1} = \hat{\mc{S}}\of{t} \cup \vec{a}_{\hat{i}}$. For SSE scheme, update $\hat{\mc{S}}\of{t+1} = \hat{\mc{S}}\of{t} \cup \vec{A}_{k(\hat{i})}$, where $k(\hat{i})$ denotes the sub-block $k$ corresponding to the identified column $\hat{i}$. 

\State \textbf{Stopping condition:} Increment the counter $t = t+1$. If $t < K$, repeat the above steps Match, Decode and Update. Else go to Step Ouptut.

\State \textbf{Output:} Recovered sparse signal is $\hvec{x}\of{K}$ and the recovered codeword $\hvec{s} = \vec{A} \hvec{x}\of{K}$. 

\end{algorithmic}
\end{algorithm}

Main computationally intensive step in MAD involves computing the correlation between the observation (or residual) and the columns of the dictionary matrix in \eqref{metri0}, which amounts to computing the matrix multiplication $\vec{A}^* \vec{y}$. When $\vec{A}$ is constructed using Gold codes (with scaled entries $\{\pm 1\}$, or power of 2 MUB matrices (with scaled entries $\{ \pm 1, \pm j \}$), the matrix multiplication $\vec{A}^* \vec{y}$ can be simply computed using only additions (and subtractions).  We also note that, the correlation of residual with the columns of dictionary matrix in \eqref{metri0} needs to be computed only for the first iteration. For the subsequent iterations, from \eqref{resi}, 
we have the recursion, $\langle \vec{r}\of{t+1} ,\vec{a}_i \rangle  = \langle \vec{r}\of{t},\vec{a}_i\> \rangle - 
b_{\hat{m}} \langle \vec{a}_{\hat{i}} ,\vec{a}_i \rangle$, where
$b_{\hat{m}}$ and $\hat{i}$ denote the symbol and the active column detected in the previous iteration. 
We can store the symmetric gram matrix $\vec{A}^* \vec{A}$, to get the values of $\langle \vec{a}_{\hat{i}} ,\vec{a}_i \rangle$ needed in the recursion. For power of 2 MUB matrices, based on the conjecture in \secref{dict}, the entries of the gram matrix will be $0$ or $1$ or from the set given in \eqref{innerProdsMUB}.

\subsection{Performance Guarantees of MAD algorithm}

Now, we establish some properties of MAD algorithm for SPARC codes. 

\begin{theorem} \label{thm:ml}
For the AWGN channel, MAD algorithm coincides with the maximum likelihood decoder of SPARC codes, when sparsity $K=1$. 
\end{theorem}
\begin{proof}
Maximum likelihood (ML) detector for AWGN finds the codeword which is the closest to the given observation, among all the possible codewords \cite{madhow2008fundamentals}. For $K=1$ SPARC code, the ML detector outputs the column index $\hat{i}$ and the modulation symbol $\hat{b}$ as
\begin{align*}
(\hat{i},\hat{b}) &= \arg\min_{ b \in \mc{M}, 1 \leq i \leq L } \| \vec{y} - b \vec{a}_i \|^2, \\
&= \arg\min_{ b \in \mc{M}, 1 \leq i \leq L } \|\vec{y}\|^2 - 2 \mathfrak{Real} \{ \langle \vec{y}, b \vec{a}_i \rangle \} + \| b \vec{a}_i \|^2, \\
&= \arg\max_{ b \in \mc{M}, 1 \leq i \leq L } \mathfrak{Real}\{ b^* \langle \vec{y}, \vec{a}_i \rangle \} - \frac{|b|^2}{2},
\end{align*}
since $\|\vec{a}_i\|^2 = 1, \forall i$. Clearly, this ML output coincides with the output of the MAD decoder (without any partial information, that is,  $\hvec{x}=\vec{0}$). 
\end{proof}

Now, we consider the recovery guarantee of the MAD decoder, in the absence of noise. Towards that, we restrict our attention to PSK constellations, $\mc{M} = \{b_1,\cdots,b_M\}$ with $|b_m| = 1, \forall m$. We define the \emph{coherence} of the PSK constellation as 
\begin{align}
\gamma &= \max_{i \neq m} \mathfrak{Real} \{b_i^* b_m \}. \label{eq:const_coh}
\end{align}
Based on the above definition, the coherence $\gamma$ for a constellation can be negative as well. Also, the coherence is not affected when a constant phase $e^{j\theta}$ is multiplied to all the symbols of the constellation. Note that, coherence $\gamma=-1$ for the BPSK constellation $\{1,-1\}$ and coherence $\gamma=0$ for the QPSK constellation $\{1,-1,j,-j\}$ (or any other rotation of the QPSK constellation). It easily follows that, the minimum distance of the PSK constellation 
$d_{\min} = \min_{i \neq m} |b_i - b_m|$ can be written in terms of its coherence as $d_{\min} = \sqrt{2 - 2\gamma}$. 

\begin{theorem}
For SPARC codes with dictinary matrix having mutual coherence $\mu$ and modulation symbols chosen from a PSK constellation having coherence $\gamma$, the MAD decoder recovers the support and modulation symbols perfectly from the noiseless observation, if the following condition is met, 
\begin{align}
K &< \min \left \{ \frac{1+\mu}{2 \mu},  \frac{1+2\mu - \gamma}{2 \mu} \right\}. \label{eq:cond12}
\end{align} 
\end{theorem}

\begin{proof}
Note that, in the first iteration, the correlation values $p_{i,m}$ computed in \eqref{metri} are identical to $p_{i,m} = \mathfrak{Real} \langle \vec{y}, b_m \vec{a_i} \rangle =  \mathfrak{Real} \{ b_m^* \vec{a}^*_i \vec{y} \}$. 
When the above conditions in \eqref{cond12} are met, we want to show that the metric corresponding to the the correct constellation symbol and the correct column (which participated in the linear combination to generate the given codeword as in \eqref{cw1}) will be higher than that of all the incorrect  cases (wrong constellation symbol and/or wrong column).  Without loss of generality (WLOG), let the codeword be generated using the first $K$ columns $\{\vec{a}_1,\cdots,\vec{a}_K\}$ columns of the dictionary matrix. For some specific column $\vec{a}_\ell$ with $1\leq \ell \leq K$, WLOG, let the modulation symbol be $b_1 \in \mc{M}$, such that, the noiseless observation is 
\begin{align*}
\vec{y} &= b_1 \vec{a}_\ell + \sum_{1\leq k \leq K, k \neq \ell} \beta_k \vec{a}_k,
\end{align*}
where $\beta_k$'s are arbitrary modulation symbols from $\mc{M}$.
The metric $p_{\ell,1}$ corresponding to an active column with correct modulation symbol is bounded as
\begin{align}
 p_{\ell,1} &= \mathfrak{Real} \langle b_1 \vec{a}_\ell + \sum_{1\leq k \leq K, k \neq \ell} \beta_k \vec{a}_k, b_1 \vec{a}_\ell \rangle, \nonumber \\
&=  \langle b_1 \vec{a}_\ell ,   b_1 \vec{a}_\ell \rangle + \mathfrak{Real} \langle \sum_{1\leq k \leq K, k \neq \ell} \beta_k \vec{a}_k, b_1 \vec{a}_\ell \rangle, \nonumber \\
&= 1 + \mathfrak{Real} \sum_{1\leq k \leq K, k \neq \ell} b_1^* \beta_k \vec{a}^*_\ell  \vec{a}_k, \nonumber \\
&\geq 1 - (K-1)\mu, \label{eq:cmet}
\end{align} 
since $|\vec{a}^*_\ell  \vec{a}_k|\leq \mu$ and $|b_1|=|\beta_k|=1$. Now, the correlation corresponding to the correct column but wrong modulation symbol $p_{\ell,m}$ with $m \neq 1$ can be bounded as,
\begin{align}
 p_{\ell,m} &= \mathfrak{Real} \langle b_m \vec{a}_\ell + \sum_{1\leq k \leq K, k \neq \ell} \beta_k \vec{a}_k, b_1 \vec{a}_\ell \rangle, \nonumber \\
&= \mathfrak{Real} \{ b_1^* b_m \} + \mathfrak{Real} \sum_{1\leq k \leq K, k \neq \ell} b_1^* \beta_k \vec{a}^*_\ell  \vec{a}_k, \nonumber \\
&\leq \gamma + (K-1)\mu, \label{eq:wmet1}
\end{align}
since $\mathfrak{Real} \{b_1^* b_m\} \leq \gamma$. Similarly, the correlation corresponding to the wrong column $p_{i,m}$ with $i>K$ can be bounded as 
\begin{align}
p_{i,m} \leq K \mu, \forall i>K, \forall m. \label{eq:wmet2}
\end{align}
Since $1\leq \ell \leq K$ is an arbitrary active column, when \eqref{cond12} is met, metric of an active column with correct symbol in \eqref{cmet} will be higher than the metrics of all the incorrect cases \eqref{wmet1} and \eqref{wmet2}. Hence MAD will find the correct column and symbol in the first iteration. After the cancellation of detected column, the problem boils down to detecting $K-1$ active columns and the corresponding symbols. 
Since the number of active columns has decreased ($K-1$ will also be less than the right hand side of the condition in \eqref{cond12}), the subsequent iterations will also be successful. 
\end{proof}
For BPSK and QPSK constellations for which $\gamma \leq 0$, condition in \eqref{cond12} simplifies as $K < \frac{1+\mu}{2\mu}$.  
Interestingly, this recovery condition coincides with that of the orthogonal matching pursuit for $K$-sparse signals \cite{tropp2004greed}. With MUB dictionary matrices, this recovery condition becomes $K < 1+\frac{\sqrt{N}}{2}$. Hence, when the sparsity level is of the order of $\sqrt{N}$, greedy algorithms can give good recovery performance. 

\subsection{Parallel MAD algorithm}
Intuitively, the first iteration of the MAD algorithm is the most error prone, since it faces the \emph{interference} from all the undetected columns. To improve on MAD performance,
we consider a variation, referred as parallel MAD. In the first iteration, we choose $T$ candidates for the active column, by taking the top $T$ metrics \eqref{metri}, and perform MAD decoding for each of these $T$ candidates, resulting in $T$ different estimates for the sparse signal. Among these $T$ estimates, we select the one with the smallest Euclidean distance to the observation, inspired by the ML decoder for white Gaussian noise. The mathematical details are described in Algorithm \ref{pmad} for completeness. The PMAD decoder is similar to parallel greedy search given in \cite{kumar2022parallel} with the notable difference that the alphabet size is discrete and the exact sparsity level is known in the current work.

\begin{algorithm}
\caption{Parallel Match and Decode Algorithm} \label{pmad}
\begin{algorithmic}[1]

\State Given the dictionary matrix $\vec{A}$ and the observation vector $\vec{y}$, 
compute $c_i = \langle \vec{y} , \vec{a}_i \rangle,~i=1,\cdots,L$ and $p_{i,m} =  
\mathfrak{Real}\{c_i b_m^*\} - \frac{|b_m|^2}{2}, ~ b_m \in \mc{M}$.

\State Initialize parallel path index $n=1$; Initialize $\mc{D} = \emptyset$.  

\State Choose a candidate for active column and the corresponding non-zero entry: $(\hat{i}_n,\hat{m}_n) = \arg\max_{(i \notin \mc{D} ,m)} p_{i,m}$.

\State  Run MAD algorithm with inputs $(\vec{A},\vec{y},K)$ and prior information on sparse signal $\hvec{x} = b_{\hat{m}_n} e_{\hat{i}_n}$. 
Denote the recovered sparse signal output of MAD as $\hvec{x}_n$. 

\State Update $\mc{D} = \mc{D} \cup \hat{i}_n$ and $n=n+1$; If $n \leq T$, go back to Step 3.

\State Final output $\bvec{x} = \arg\min_{\hvec{x}_n; 1 \leq n \leq T} \| \vec{y}-\vec{A}\hvec{x}_n\|$.

\end{algorithmic}
\end{algorithm}

\section{Application to Multi-User Channels} \label{sec:mult}
In this Section, we discuss how the SSE and PMAD can be used in multi-user scenarios. Using the SSE scheme with sparsity level $K$ described in \secref{sse}, we can support $P$-user multiple access channel, or $P$-user broadcast channel or $P$-user interference channel \cite{cover1999elements,el2011network}, for any $P \leq K$. First, we illustrate how the SSE scheme can be employed to generate the codeword of each user based on the user's information bits. As before, we partition the dictionary matrix $\vec{A}$ into $K$ sub-blocks, with sub-block $\vec{A}_k$ having $L_k$ number of columns. These $K$ sub-blocks are divided among $P$ users, with $\mc{A}_i = \{\vec{A}_{i,1},\cdots,\vec{A}_{i,K_i}\} \subset \{\vec{A}_1,\cdots,\vec{A}_K\}$ denoting the ordered set of $K_i$ sub-blocks assigned to user-$i$. 
Note that $\mc{A}_i \cap \mc{A}_j = \emptyset$ if $i \neq j$ and $\displaystyle \sum_{i=1}^K K_i = K$. The codeword for user-$i$ is obtained as
\begin{equation}
\vec{s}_i = \sum_{k=1}^{K_i} \beta_{i,k} \vec{a}_{i,k} \label{eq:cwi}
\end{equation}
where symbols $\{\beta_{i,1},\cdots,\beta_{i,K_i}\}$ are chosen from $M_i$-ary constellation and the column $\vec{a}_{i,k}$ is chosen from the sub-block $\vec{A}_{i,k}$, for $1 \leq k \leq K_i$. Denoting the number of columns in $\vec{A}_{i,k}$ as $L_{i,k}$, the total number of bits that can be conveyed for user-$i$ is 
\begin{equation}
N_{b_i} = K_i \log M_i  + \sum_{i=1}^{K_i} \log L_{i,k}. \label{eq:nbi} 
\end{equation}

In the multiple access channel (MAC), which is equivalent to an uplink scenario in a cellular network, the encoding is done independently by each user, which coincides with the SSE based procedure in \eqref{cwi}. The observation at the receiver is 
\begin{eqnarray}
\vec{y} &=& \sum_{i=1}^P \vec{s}_i + \vec{v}  ~=~ \sum_{i=1}^P \sum_{k=1}^{K_i} \beta_{i,k} \vec{a}_{i,k} + \vec{v}. \label{eq:mac}
\end{eqnarray}
The decoding is done jointly at the receiver, which can be done using the MAD or PMAD algorithm, which recovers the support of the active columns $\{\vec{a}_{i,k}\}$ and the corresponding modulation symbols $\beta_{i,k}$, for each user. 

In the broadcast channel, which is similar to the downlink scenario in a cellular network, encoding is done jointly at the base station and the decoding is done by each user separately. 
The transmitter sends the sum of all the users' codewords as 
$\vec{s} = \sum_{i=1}^P \vec{s}_i$. 
Received signal at the user-$i$ is given by 
$\vec{y}_i = \vec{s} + \vec{n_i}$,
where $\vec{n}_i$ is the noise at the user-$i$. 
MAD or PMAD decoding can be employed by each user, which recovers the active columns present in $\vec{s}$ and the corresponding modulation symbols. Hence, in this approach, users recover the information sent to the other users, in addition to their own information. If the users are grouped such that their noise statistics are similar, then their error performance will be similar. We also note that SSE can be applied in two-way relay channel, which has a multiple access phase followed by a broadcast phase. 

In the interference channel, there are $P$ transmitters and $P$ receivers. Each transmitter sends information to a corresponding intended receiver. With $i^{th}$ transmitter generating codeword as in \eqref{cwi}, the received signal at the $i^{th}$ receiver is given by
\begin{equation}
\vec{y}_i = \vec{s}_i + \sum_{j \neq i} h_{i,j} \vec{s}_j + \vec{n}_i, \label{eq:ic}
\end{equation}
where $h_{i,j}$ denotes the channel gain from $j^{th}$ transmitter to the $i^{th}$ receiver.  Without loss of generality, we have taken $h_{i,i} = 1$. 
MAD or PMAD decoding employed at the $i^{th}$ receiver recovers the codewords of all the transmitters. If $|h_{i,j}| = 1, \forall i,j$, and the noise statistics are identical across all the receivers, then the decoding performance (successful recovery of all the codewords) of all the receivers will coincide with the corresponding single user case.



\section{Simulation Results} \label{sec:simu}
We study the performance in terms of the block error rate (BLER), also referred as codeword error rate, for the proposed encoding and decoding schemes in additive white Gaussian noise channels. For the complex MUB dictionary matrix, the non-zero entries of the sparse signal are chosen from the QPSK constellation. For real Gold code dictionary matrix, we consider BPSK constellation. When the non-zero entries in the $K$-sparse signal $\vec{x}$ are uncorrelated, it easily follows that the expected energy of the codeword $\vec{s}=\vec{A}\vec{x}$ is $E_s = K$, with the columns of dictionary matrix being unit norm. Energy per bit $E_b$ is obtained by dividing $E_s$ by the total number of bits $N_b$ conveyed by the sparse signal $\vec{x}$. With $\frac{N_0}{2}$ denoting the variance of the Gaussian noise per real dimension, we study the BLER versus $E_b/N_0$ (in dB) of the proposed schemes.

An error control code conveying $N_b$ bits using $N$ real channel uses is represented by the pair $(N,N_b)$, with the code rate of $\frac{N_b}{N}$ bits per real channel use. A complex code of length $N$ can be represented by a real code of length $2N$ by concatenating real and imaginary part of the code.
In this paper, a complex code of length $N$ supporting $N_b$ bits of information is equivalent to a $(2N,N_b)$ real code, with code rate of $\frac{N_b}{2N}$ bits per real channel use.

\subsubsection{MAD vs. OMP} Orthogonal Matching Pursuit (OMP) decoder is a well studied greedy decoder \cite{tropp2004greed} for sparse signal recovery, which is similar in computational cost to that of MAD decoder.  In Figure~\ref{fig:MADvsOMP}, we compare the BLER performance of both decoders for complex MUB dictionary of size $64\times4096$ and sparsity $K=6$ giving rise to a $(128,68)$ SSE-GSPARC code. 

We run the OMP algorithm for $K$ iterations and quantize the non-zero entries of the recovered $K$-sparse signal (obtained using least squares method) to the nearest nearest constellation points.  On the other hand, MAD algorithm utilizes the finite alphabet size of the non-zero entries in every iteration, by jointly decoding the active column and the corresponding constellation point. In addition, when OMP projects the residuals onto the orthogonal complement of the detected columns, there will be a reduction of signal components from the yet-to-be detected active columns. On the other hand, MAD simply subtracts out the detected columns without affecting the yet-to-be detected active columns. Due to these reasons, the proposed MAD decoder provides better BLER performance than the OMP algorithm.

In addition to the standard QPSK constellation $\{+1,-1,+j,-j\}$, we also consider offset QPSK constellations. Specifically, the modulating symbol for $k^{th}$ sub-block for $k\in\{1,...,K\}$ is chosen from a rotated QPSK constellation, obtained by counter-clockwise rotation of the standard constellation by $\displaystyle \frac{(k-1)\pi}{2K}$ radians. The motivation for introducing phase offset to different sub-blocks is based on the following reasoning. With $i$ being an index of one of the active columns from the sparse signal support set $\mc{S}$, consider the inner product $\langle \vec{y},\vec{a}_i \rangle = \beta_i + \sum_{k \in \mc{S}, k \neq i} \beta_k \langle \vec{a}_k,\vec{a}_i \rangle  + \langle \vec{v},\vec{a}_i\rangle $. MAD decoder is prone to error when the net interference from other active columns has high magnitude. For the complex MUB dictionary matrix with $N=64$, from \eqref{innerProdsMUB}, the inner product between any two non-orthogonal columns belong to the set $\{\frac{+1}{\sqrt{N}},\frac{-1}{\sqrt{N}},\frac{+j}{\sqrt{N}},\frac{-j}{\sqrt{N}}\}$. Due to this property, there are many possible support sets $\mc{S}$, for which the interference terms can add coherently to result in a high magnitude. To mitigate this constructive addition of interfering terms, we introduce a phase offset to each sub-block of the dictionary matrix. From the results in Fig.~\ref{fig:MADvsOMP}, we see that MAD algorithm performs better with offset QPSK constellations. In all the remaining plots for SSE schemes with complex MUB dictionary matrix, we have used offset QPSK as the default modulation scheme. 


\begin{figure}
    \centering
    \includegraphics[width=\linewidth]{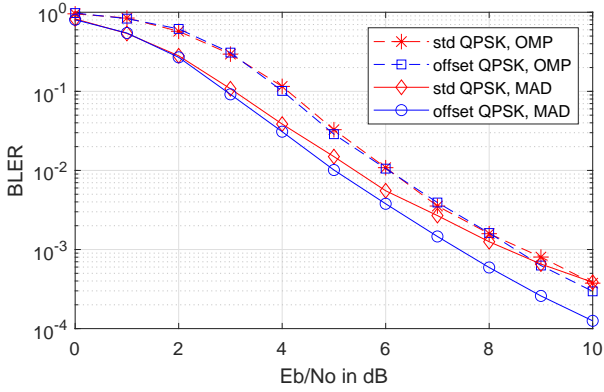}
    \caption{Comparison of OMP and MAD decoder for SSE schemes.}  
    \label{fig:MADvsOMP}
\end{figure}

\subsubsection{MAD vs. PMAD} The performance of MAD decoder can be improved by running multiple MAD decoders in parallel and selecting the best solution based on minimum distance decoding rule. In Figure~\ref{fig:MADvsPMAD}, we compare MAD decoder with PMAD decoder, for different number of parallel paths. The simulation parameters are the same as in Fig.~\ref{fig:MADvsOMP}, with offset QPSK constellations. We denote PMAD with $T$ parallel paths as $T-$PMAD. We observe that the PMAD improves the BLER performance of MAD decoder significantly. 16-PMAD and 100-PMAD has roughly 4.5 dB and 5 dB gain respectively over MAD decoder for BLER of $10^{-4}$. This shows that the gains from parallel search starts to diminish as we increase the number of parallel paths. This allows us to use a small number of parallel paths for our PMAD decoders.

\begin{figure}
    \centering
    \includegraphics[width=\linewidth]{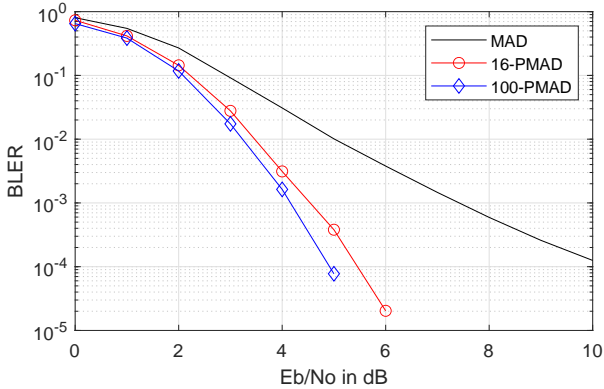}
    \caption{Effects of number of parallel paths on the BLER of PMAD decoder.}
    \label{fig:MADvsPMAD}
\end{figure}

\subsubsection{PMAD vs. AMP} Approximate Message Passing decoders have been developed for standard SPARC and modulated SPARC for random Gaussian dictionary matrices \cite{hsieh2021modulated} and have been shown empirically to work with other dictionary matrices. In Figure~\ref{fig:PMADvsAMP}, we compare the BLER performances of PMAD decoder with online AMP decoder \cite{hsieh2021modulated}. 
We use equal power allocation for all the sub-blocks, because the power allocation techniques to improve the performance of AMP decoders do not work  in the small code length $(N\leq128)$ and low code rate $(R\approx0.5)$ regime  \cite{greig2017techniques,liang2021finite}. 
We compare PMAD algorithm with $T$ parallel paths with the AMP with $T$ iterations, which is referred as $T-$AMP in the plot. The AMP algorithm \cite{hsieh2021modulated} computes non-linear MMSE estimate of each entry of the sparse signal $\vec{x}$ in each iteration. We note that the $T-$PMAD requires significantly less computations than $T-$AMP. 

We consider two scenarios, one with equal sub-block sizes and the other with unequal size sub-blocks. A complex MUB dictionary matrix of size $64\times512$ with $K=8$ has equal size sub-blocks, resulting in a $(128,64)$ code.   With complex MUB dictionary matrix of size $64\times4096$, running the sub-block partitioning algorithm from \secref{spa} with $K=6$, we get a $(128,68)$ code with unequal sub-block sizes. Since AMP in \cite{hsieh2021modulated} is designed for equal size sub-blocks, we use a generalization of the AMP to accommodate unequal sub-block sizes. From Fig.~\ref{fig:PMADvsAMP}, we find that $T$-PMAD performs better than $T$-AMP, in the short block length regime. For the unequal size sub-blocks, AMP performs poorly for large values of $E_b/N_0$. However, PMAD algorithm works well for both equal and unequal sub-block sizes. We also note that, the lower sparsity case $K=6$ with code rate $\frac{68}{128}$ performs better than the higher sparsity case $K=8$ with code rate $0.5$, emphasizing that sparsity is a key parameter for SPARC.     


\begin{figure}
    \centering
    \includegraphics[width=\linewidth]{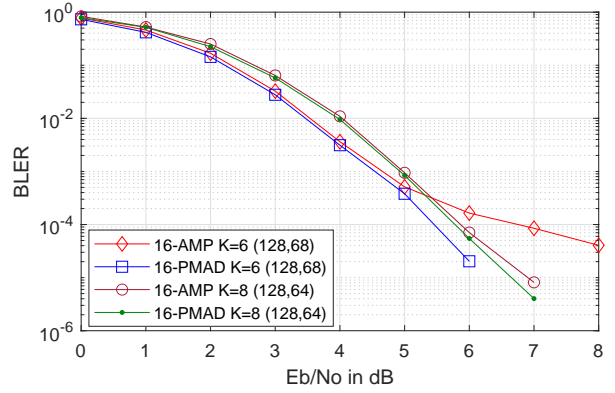}
    \caption{Comparison of PMAD decoder with online-AMP decoder.}
    \label{fig:PMADvsAMP}
\end{figure}

\subsubsection{SSE vs. SFE} In Figure~\ref{fig:SSEvsSFE}, we study the BLER performance of the encoding schemes with and without sub-block structure, for a complex MUB dictionary matrix of size $64\times4096$, using 100-PMAD decoder. SSE with $K=5$ and $K=6$ gives $(128,58)$ and $(128,68)$ codes respectively, while SFE with $K=5$ gives $(128,63)$ code. Since SFE transmits more bits than SSE for the same sparsity $K$, the noise level in SFE will be smaller than that of the SSE scheme. On the other hand, the search space for each iteration of the greedy decoder for SFE will be larger than that of the SSE scheme. Due to these counteracting effects, SFE has nearly same BLER performance as SSE, at high $E_b/N_0$ values, while achieving higher code rate.  

\begin{figure}
    \centering
    \includegraphics[width=\linewidth]{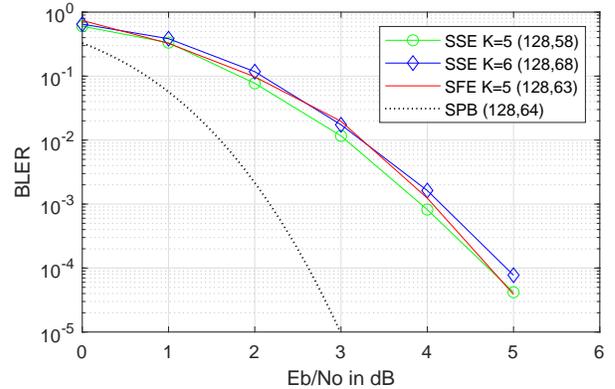}
    \caption{Comparison of BLER performance of SSE and SFE schemes.}
    \label{fig:SSEvsSFE}
\end{figure}

\subsubsection{Very short length codes} In Figure~\ref{fig:veryShort}, we GSPARC using MAD/PMAD decoding for very short lengths, with $(20,11)$ and $(20,8)$ Golay codes considered for 5G-NR \cite{van2018short}, using ML decoding. Complex MUB dictionary matrix of size $8\times64$ with $K=1$ gives $(16,8)$ code, for which MAD decoder is used. Complex MUB dictionary matrix of size $16 \times 257$ with SFE $K=2$ scheme gives $(32,19)$ code, for which 16-PMAD decoder is used. We find that GSPARC codes give comparable performance to Golay codes of very short lengths.

\begin{figure}
    \centering
    \includegraphics[width=\linewidth]{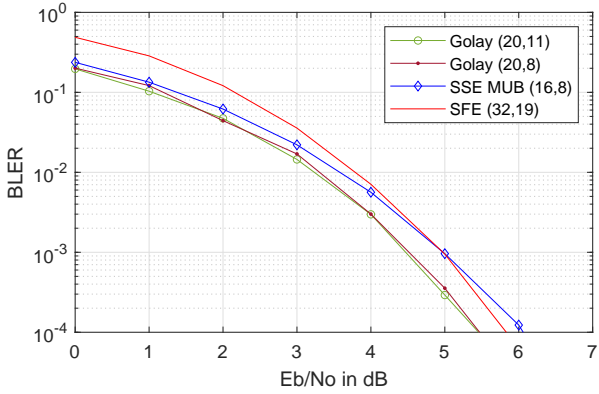}
    \caption{Comparison of BLER performance of very short length codes.}
    \label{fig:veryShort}
\end{figure}

\subsubsection{Short length codes} In Figure~\ref{fig:shortcodes}, we compare our $(127,63)$ SSE scheme (Gold code dictionary matrix of size $127\times128^2$ with $K=5$) with some of the existing $(128,64)$ error control codes \cite{cocskun2019efficient}: binary LDPC codes used in the CCSDS standard,  LDPC codes (base graph 2) considered for 5G-NR standard and Turbo code with 16 states.  More details about these existing codes are given in \cite{cocskun2019efficient}. SSE with PMAD decoder performs better than LDPC code from the CCSDS standard. We also note that some codes like tail-biting convolutional code  with constraint length 14 \cite{cocskun2019efficient} and polarization adjusted convolutional codes \cite{arikan2019sequential} perform very close to sphere packing bounds (shown in Fig.~\ref{fig:shortcodes} with legend 'SPB') for the given code length and code rate.

\begin{figure}
    \centering
    \includegraphics[width=\linewidth]{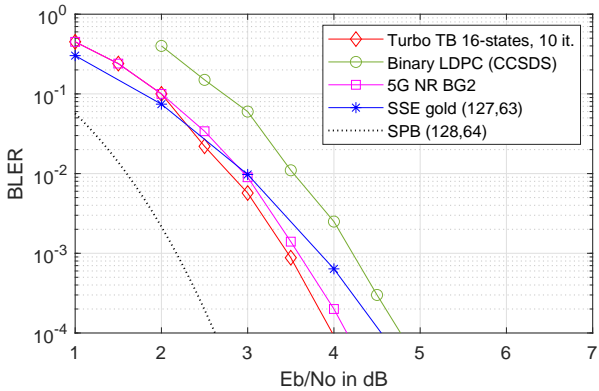}
    \caption{Comparison with existing $(128,64)$ codes.}
    \label{fig:shortcodes}
\end{figure}

\subsubsection{Multi-user channels} As explained in \secref{mult}, SSE with sparsity $K$ can support up to $K$ users in multi-user channels. For illustration, we consider a multiple-access channel \eqref{mac}. An SSE scheme with sparsity level $K$ resulting in a $(N_1,N_b)$ code  utilizes $N_1$ real channel uses and communicates either $\displaystyle \lfloor N_b/K \rfloor$ or $\displaystyle \lceil N_b/K \rceil$ bits from each user, based on the optimal sub-block partitioning algorithm from \secref{spa}. For comparison, we consider a $K$-user orthogonal multiple access scheme, where each user is assigned a dedicated time/frequency resource and each user employs a single user Golay code with approximate parameters $(N_1/K, N_b/K)$. We also find the lower bound for the $K$-user orthogonal multiple access using sphere packing bounds for code parameters $(N_1/K, N_b/K)$. In Figure~\ref{fig:multiuser}, we study the probability that at least one user is decoded in error. For $K=6$ users, using SSE with Gold code dictionary matrix resulting in $(127,74)$ code (communicating 12 or 13 bits for each user) outperforms the sphere packing bounds of the orthogonal multiple access scheme using $(23,12)$ codes, and also provides higher spectral efficiency. Similar results hold true for MUB dictionary matrices as well. We see that, SSE with PMAD provides a multi-user error control coding scheme, offering significant gains over orthogonal multiple access schemes, for short block lengths. The gains can be understood from the fact that SSE encodes the information from the users over block length of $N_1$, while orthogonal multiple access schemes use codes of smaller block lengths $N_1/K$. SSE provides a neat way of pooling the resources of users together such that the overall error performance of all the users is improved.

\begin{figure}
    \centering
    \includegraphics[width=\linewidth]{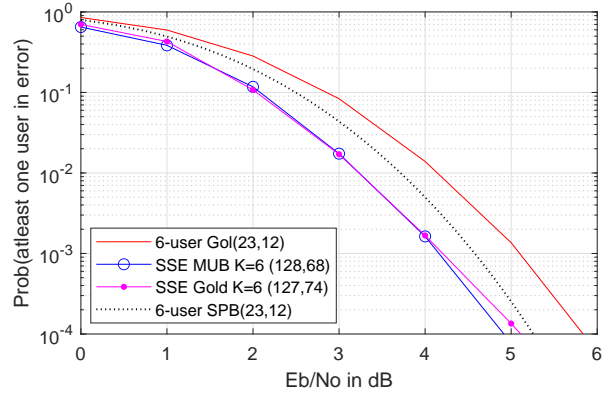}
    \caption{Comparison of BLER performance in the multiple access channel.}
    \label{fig:multiuser}
\end{figure}

\section{Conclusions} \label{sec:conc}
In this paper, we developed two generalizations of SPARC, an SSE scheme, which allows unequal sub-block sizes and an SFE scheme, which eliminates the sub-block structure altogether. For both SSE and SFE schemes, we developed a greedy approach based decoder, referred as MAD algorithm and introduced a parallel greedy search mechanism to improve its performance. Using Gold codes and mutually unbiased bases to construct the dictionary matrices, we study block error rate performance in AWGN channels, for short block lengths.  We showed that our proposed PMAD outperforms the AMP decoder for SPARC and performs comparably and competitively with widely used codes. We also described that SSE scheme can be used in various multi-user channel settings. In multiple access channels, we showed that SSE with PMAD decoder outperforms the sphere packing lower bounds of an orthogonal multiple access scheme, having the same spectral efficiency. Developing greedy decoders for GSPARC which work well for moderate to large block lengths can be explored in a future work. 
Studying GSPARC in multi-user channels with asymmetric power and rate conditions can be explored in the future. 

\appendix

\subsection{Indexing the ordered set of $K$ objects out of $L$ objects} \label{app:index_comb}

Given $L$, $K$ and $d$, our goal is to find a good estimate for $i$, which satisfies the condition \eqref{require}. The condition can be rewritten as 
\begin{align*}
    \binom{L}{K}\left(1-\frac{\binom{L-i}{K}}{\binom{L}{K}}\right)
    & \leq d
    <\binom{L}{K}\left(1-\frac{\binom{L-(i+1)}{K}}{\binom{L}{K}}\right), \\
    \implies
    \frac{\binom{L-(i+1)}{K}}{\binom{L}{K}}
    &< 1-\frac{d}{\binom{L}{K}}
    \leq \frac{\binom{L-i}{K}}{\binom{L}{K}}.
\end{align*}
    
First, noting that
\begin{align}
 \frac{L-m-p}{L-p} &= 1-\frac{m}{L-p} ~~ \leq 1-\frac{m}{L} ~~= \frac{L-m}{L},
   \label{eq:RatioBound}
\end{align}
for $p\in\{0,1,...,K-1\}$, we get the following bounds
\begin{align}
     \left(\frac{L-(m+1)-(K-1)}{L-(K-1)}\right)^K  
     & \leq \frac{\binom{L-(m+1)}{K}}{\binom{L}{K}},  \nonumber \\
     \frac{\binom{L-m}{K}}{\binom{L}{K}}
     &\leq \left(\frac{L-m}{L}\right)^K  \label{eq:Bound1}
\end{align}
%
Setting $\Bar{L}=L-\frac{K-1}{2}$, we have
\begin{align*}
    \frac{\binom{L-m}{K}}{\binom{L}{K}}
    &= \frac{(\Bar{L}+\frac{K-1}{2}-m) 
    \cdots(\Bar{L}+\frac{K-1}{2}-m-(K-1))}{(\Bar{L}+\frac{K-1}{2})
    \cdots(\Bar{L}+\frac{K-1}{2}-(K-1))} \\
    &= \left\{ \begin{array}{ll} 
        \frac{\Bar{L}-m}{\Bar{L}}\times\prod_{p=1}^{\frac{K-1}{2}}\frac{(\Bar{L}-m)^2-p^2}{\Bar{L}^2-p^2} & \text{for odd $K$,} \\
        \prod_{p=1}^{\frac{K}{2}}\frac{(\Bar{L}-m)^2-\left(\frac{2p-1}{2}\right)^2}{\Bar{L}^2-\left(\frac{2p-1}{2}\right)^2} &
        \text{for even $K$.} \end{array} \right.       
\end{align*}


For both odd and even $K$, we have
\begin{equation}
 \left(\frac{(\Bar{L}-m)^2-\left(\frac{K-1}{2}\right)^2}{\Bar{L}^2-\left(\frac{K-1}{2}\right)^2}\right)^\frac{K}{2}
  \leq
  \frac{\binom{L-m}{K}}{\binom{L}{K}}
  \leq\left(\frac{\Bar{L}-m}{\Bar{L}}\right)^K
\end{equation}
Following the arguments of equation \ref{eq:RatioBound}, it is easy to show the following inequality:
\begin{align}
&\left(\frac{L-m-(L-1)}{L-(K-1)}\right)^K
\leq \left(\frac{(\Bar{L}-m)^2-\left(\frac{K-1}{2}\right)^2}{\Bar{L}^2-\left(\frac{K-1}{2}\right)^2}\right)^\frac{K}{2} \nonumber \\
&~~~~~\leq \frac{\binom{L-m}{K}}{\binom{L}{K}} 
\leq \left(\frac{\Bar{L}-m}{\Bar{L}}\right)^K
\leq \left(\frac{L-m}{L}\right)^K.
\label{eq:Bound2}
\end{align}
Combining inequalities from \eqref{Bound1} and \eqref{Bound2}, we get 
\begin{equation*}
    \left(\frac{(\Bar{L}-(i+1))^2-\left(\frac{K-1}{2}\right)^2}{\Bar{L}^2-\left(\frac{K-1}{2}\right)^2}\right)^\frac{K}{2}
    < 1-\frac{d}{\binom{L}{K}}
    \leq\left(\frac{\Bar{L}-i}{\Bar{L}}\right)^K,
\end{equation*}
from which, 
we get the upper bound for $i$ in \eqref{ibar}. 






%


\bibliographystyle{ieeetr}
\bibliography{ref}

\begin{thebibliography}{10}

\bibitem{shannon1948mathematical}
C.~E. Shannon, ``A mathematical theory of communication,'' {\em The Bell system
  technical journal}, vol.~27, no.~3, pp.~379--423, 1948.

\bibitem{lin2001error}
S.~Lin and D.~J. Costello, {\em Error control coding}.
\newblock Prentice hall, 2001.

\bibitem{van2018short}
J.~Van~Wonterghem, A.~Alloum, J.~J. Boutros, and M.~Moeneclaey, ``On
  short-length error-correcting codes for {5G-NR},'' {\em Ad Hoc Networks},
  vol.~79, pp.~53--62, 2018.

\bibitem{cocskun2019efficient}
M.~C. Co{\c{s}}kun, G.~Durisi, T.~Jerkovits, G.~Liva, W.~Ryan, B.~Stein, and
  F.~Steiner, ``Efficient error-correcting codes in the short blocklength
  regime,'' {\em Physical Communication}, vol.~34, pp.~66--79, 2019.

\bibitem{joseph2012least}
A.~Joseph and A.~R. Barron, ``Least squares superposition codes of moderate
  dictionary size are reliable at rates up to capacity,'' {\em IEEE
  Transactions on Information Theory}, vol.~58, no.~5, pp.~2541--2557, 2012.

\bibitem{joseph2013fast}
A.~Joseph and A.~R. Barron, ``Fast sparse superposition codes have near
  exponential error probability for $r<c$,'' {\em IEEE Transactions on
  Information Theory}, vol.~60, no.~2, pp.~919--942, 2014.

\bibitem{eldar2012compressed}
Y.~C. Eldar and G.~Kutyniok, {\em Compressed sensing: theory and applications}.
\newblock Cambridge university press, 2012.

\bibitem{rush2017capacity}
C.~Rush, A.~Greig, and R.~Venkataramanan, ``Capacity-achieving sparse
  superposition codes via approximate message passing decoding,'' {\em IEEE
  Transactions on Information Theory}, vol.~63, no.~3, pp.~1476--1500, 2017.

\bibitem{hsieh2021modulated}
K.~Hsieh and R.~Venkataramanan, ``Modulated sparse superposition codes for the
  complex awgn channel,'' {\em IEEE Transactions on Information Theory},
  vol.~67, no.~7, pp.~4385--4404, 2021.

\bibitem{cho2013approximate}
S.~Cho and A.~Barron, ``Approximate iterative bayes optimal estimates for
  high-rate sparse superposition codes,'' in {\em Sixth Workshop on
  Information-Theoretic Methods in Science and Engineering}, 2013.

\bibitem{greig2017techniques}
A.~Greig and R.~Venkataramanan, ``Techniques for improving the finite length
  performance of sparse superposition codes,'' {\em IEEE Transactions on
  Communications}, vol.~66, no.~3, pp.~905--917, 2017.

\bibitem{barbier2015approximate}
J.~Barbier, C.~Sch{\"u}lke, and F.~Krzakala, ``Approximate message-passing with
  spatially coupled structured operators, with applications to compressed
  sensing and sparse superposition codes,'' {\em Journal of Statistical
  Mechanics: Theory and Experiment}, vol.~2015, no.~5, p.~P05013, 2015.

\bibitem{barbier2017approximate}
J.~Barbier and F.~Krzakala, ``Approximate message-passing decoder and capacity
  achieving sparse superposition codes,'' {\em IEEE Transactions on Information
  Theory}, vol.~63, no.~8, pp.~4894--4927, 2017.

\bibitem{barbier2019universal}
J.~Barbier, M.~Dia, and N.~Macris, ``Universal sparse superposition codes with
  spatial coupling and gamp decoding,'' {\em IEEE Transactions on Information
  Theory}, vol.~65, no.~9, pp.~5618--5642, 2019.

\bibitem{liang2021finite}
S.~Liang, B.~Bai, and G.~Zhang, ``On the finite length performance of sparse
  regression codes with peak-power limitation,'' in {\em 2020 IEEE Information
  Theory Workshop (ITW)}, pp.~1--5, IEEE, 2021.

\bibitem{Goldcode}
R.~{Gold}, ``Optimal binary sequences for spread spectrum multiplexing
  (corresp.),'' {\em IEEE Transactions on Information Theory}, vol.~13, no.~4,
  pp.~619--621, 1967.

\bibitem{wootters1989optimal}
W.~K. Wootters and B.~D. Fields, ``Optimal state-determination by mutually
  unbiased measurements,'' {\em Annals of Physics}, vol.~191, no.~2,
  pp.~363--381, 1989.

\bibitem{frank1963polyphase}
R.~Frank, ``Polyphase codes with good nonperiodic correlation properties,''
  {\em IEEE Transactions on Information Theory}, vol.~9, no.~1, pp.~43--45,
  1963.

\bibitem{chu1972polyphase}
D.~Chu, ``Polyphase codes with good periodic correlation properties
  (corresp.),'' {\em IEEE Transactions on information theory}, vol.~18, no.~4,
  pp.~531--532, 1972.

\bibitem{renes2004symmetric}
J.~M. Renes, R.~Blume-Kohout, A.~J. Scott, and C.~M. Caves, ``Symmetric
  informationally complete quantum measurements,'' {\em Journal of Mathematical
  Physics}, vol.~45, no.~6, pp.~2171--2180, 2004.

\bibitem{mallat1993matching}
S.~G. Mallat and Z.~Zhang, ``Matching pursuits with time-frequency
  dictionaries,'' {\em IEEE Transactions on signal processing}, vol.~41,
  no.~12, pp.~3397--3415, 1993.

\bibitem{cai2011orthogonal}
T.~T. Cai and L.~Wang, ``Orthogonal matching pursuit for sparse signal recovery
  with noise,'' {\em IEEE Transactions on Information theory}, vol.~57, no.~7,
  pp.~4680--4688, 2011.

\bibitem{tropp2004greed}
J.~A. Tropp, ``Greed is good: {A}lgorithmic results for sparse approximation,''
  {\em IEEE Transactions on Information theory}, vol.~50, no.~10,
  pp.~2231--2242, 2004.

\bibitem{chen2001atomic}
S.~S. Chen, D.~L. Donoho, and M.~A. Saunders, ``Atomic decomposition by basis
  pursuit,'' {\em SIAM review}, vol.~43, no.~1, pp.~129--159, 2001.

\bibitem{beck2009fast}
A.~Beck and M.~Teboulle, ``A fast iterative shrinkage-thresholding algorithm
  for linear inverse problems,'' {\em SIAM journal on imaging sciences},
  vol.~2, no.~1, pp.~183--202, 2009.

\bibitem{messageMontanari}
M.~Bayati and A.~Montanari, ``The dynamics of message passing on dense graphs,
  with applications to compressed sensing,'' in {\em 2010 IEEE International
  Symposium on Information Theory}, pp.~1528--1532, 2010.

\bibitem{madhow2008fundamentals}
U.~Madhow, {\em Fundamentals of digital communication}.
\newblock Cambridge University Press, 2008.

\bibitem{kumar2022parallel}
R.~Kumar, M.~K. Sinha, and A.~P. Kannu, ``Parallel greedy search for random
  access in wireless networks,'' {\em IETE Technical Review}, pp.~1--10, 2022.

\bibitem{cover1999elements}
T.~M. Cover, {\em Elements of information theory}.
\newblock John Wiley \& Sons, 1999.

\bibitem{el2011network}
A.~El~Gamal and Y.-H. Kim, {\em Network information theory}.
\newblock Cambridge university press, 2011.

\bibitem{arikan2019sequential}
E.~Ar{\i}kan, ``From sequential decoding to channel polarization and back
  again,'' {\em arXiv preprint arXiv:1908.09594}, 2019.

\end{thebibliography}

\end{document}